\documentclass[12pt,a4paper]{article}

\usepackage{amsmath,amssymb,amsthm}
\usepackage{algorithm}
\usepackage{algorithmic}
\usepackage{booktabs}
\usepackage{tabularx}
\usepackage{multirow}
\usepackage{graphicx}
\usepackage{xcolor}
\usepackage{colortbl}
\usepackage{url}
\usepackage{cite}
\usepackage{placeins}
\usepackage[margin=1in]{geometry}

\newtheorem{definition}{Definition}
\newtheorem{proposition}{Proposition}

\title{\Large\textbf{Exact Mining of Dense Patterns via Direct Evaluation\\of Local Interval Frequency Using a Sliding Window}}
\author{
  Taihei Takahashi$^1$ \and
  Kanata Takayasu$^1$ \and
  Satoshi Suga$^2$ \and
  Satoshi Kurihara$^1$ \\[0.25em]
  \small $^1$Keio University \quad
  $^2$Kansai University
}
\date{}

\begin{document}
\maketitle

\begin{abstract}
Accurately extracting patterns that appear frequently only within specific time intervals, together with their dense intervals, is important in many applications such as understanding seasonal demand and detecting anomalous behavior.
Frequent itemset mining evaluates support over the entire dataset and therefore cannot detect locally dense patterns.
Existing methods for dense pattern mining with interval output estimate dense intervals through occurrence-gap constraints; however, since the gap constraint parameter governs both pattern identification accuracy and interval detection accuracy simultaneously, finding a parameter setting that achieves high accuracy for both objectives is difficult.
In this paper, we propose Apriori-window, an exact algorithm that resolves this structural limitation.
The proposed method directly evaluates local frequency within a sliding window and thus requires no gap constraint parameter, and it efficiently enumerates dense intervals through anti-monotonicity-based pruning of the search space and stride-skip reduction of the number of window scans.
Experiments on three real-world datasets demonstrate that existing methods struggle to simultaneously achieve high accuracy in both pattern identification and dense interval detection, and scalability experiments on synthetic data confirm the practical applicability of the proposed method.
\end{abstract}


\section{Introduction}
\label{sec:intro}

In diverse domains including retail, healthcare, and cybersecurity, large volumes of data are accumulated daily.
The content recorded in such data varies over time due to trends, seasonal fluctuations, and the influence of external events.
Extracting temporal characteristics from such data is important for understanding the nature of the data.
Pattern mining encompasses numerous methods for extracting useful information from data.
Frequent itemset mining (FIM)~\cite{Agrawal1994,Zaki1997,Han2000} is a representative approach in pattern mining that extracts combinations of items appearing frequently across the entire dataset as patterns.
However, since this approach evaluates support over the entire dataset, it cannot account for locally varying occurrence frequencies.
As a result, setting the minimum support threshold high causes patterns that are densely concentrated within limited intervals to be excluded, while setting it low causes many non-dense patterns to be extracted alongside the dense ones.
This dilemma is known as the rare item problem~\cite{Liu1999}.

To address this problem, methods that simultaneously extract dense patterns---itemsets that exhibit high occurrence frequency only within specific intervals---together with their dense intervals have been proposed.
LPFIM~\cite{Mahanta2005}, LPPM~\cite{FournierViger2021}, and RPM~\cite{Kiran2015} all estimate dense intervals through occurrence-gap-based parameters, outputting patterns and their intervals simultaneously.
However, in these methods the occurrence-gap parameter governs both pattern identification accuracy and interval detection accuracy, making it extremely difficult to find a parameter setting that maintains both at a high level simultaneously.

In this paper, we propose Apriori-window, an exact algorithm that overcomes this limitation.
The proposed method evaluates local frequency by directly counting occurrences within a sliding window, requiring no gap threshold and thus avoiding the parameter sensitivity problem.
Combining anti-monotonicity-based pruning of pattern candidates and candidate intervals with stride adjustment, the method achieves comprehensive and exact extraction in practical execution time.

\smallskip
\noindent This paper addresses the following research questions:
\begin{description}
  \item[RQ1] Is it possible to develop an exact algorithm that precisely identifies dense patterns and detects their dense intervals?
  \item[RQ2] Can existing methods accurately identify dense patterns and detect their dense intervals?
  \item[RQ3] How do the execution time and memory usage of the exact dense pattern algorithm scale with data characteristics?
\end{description}

\smallskip
\noindent The main contributions of this paper are as follows:
\begin{enumerate}
  \item We propose Apriori-window, an exact algorithm based on direct evaluation of local frequency via a sliding window.
        We prove the safety of the anti-monotonicity and the stride-skip mechanism, establishing the exactness of the extraction results (RQ1).
  \item Through parameter grid search of existing methods on three real-world datasets, we demonstrate that existing methods struggle to simultaneously achieve high accuracy in both dense pattern identification and dense interval detection (RQ2).
  \item Scalability experiments on synthetic data quantify the growth of execution time and memory usage with respect to data characteristics, confirming the practical applicability of the proposed method (RQ3).
\end{enumerate}

\section{Related Work}
\label{sec:related}

\subsection{Frequent Itemset Mining}
\label{sec:rel_fim}

Frequent itemset mining (FIM) is the task of discovering itemsets whose support over the entire dataset meets a user-specified threshold~\cite{Chee2019}.
Apriori~\cite{Agrawal1994} introduced a candidate-generation and pruning framework exploiting anti-monotonicity; ECLAT~\cite{Zaki1997} improved efficiency through a vertical data representation; and FP-Growth~\cite{Han2000} achieved further speedup via a compressed tree structure that eliminates candidate generation.
Surveys of these foundational methods and their extensions appear in~\cite{FournierViger2017survey,FournierViger2022challenges,Luna2019review}.

Extensions of the FIM framework include high-utility itemset mining~\cite{Ahmed2009,FournierViger2019hupm}, which replaces frequency with utility values such as profit.
To improve computational efficiency or pattern quality, various approaches have also been proposed: direct rule extraction using autoencoders~\cite{Berteloot2025,Karabulut2024}, discovery of low-redundancy pattern sets via binarized neural networks~\cite{Fischer2021}, adaptive search using reinforcement learning~\cite{Ghosh2025}, extraction of sequential and high-utility patterns via deep learning~\cite{Jamshed2020,Porwal2026}, and approximate search using evolutionary computation~\cite{Djenouri2018,Telikani2020}.
Since all of these methods evaluate support over the entire dataset, they cannot detect itemsets that occur at high frequency only within specific time intervals.
Patterns that are globally infrequent but locally frequent are missed when the minimum support threshold is set high, and when it is set low, a large number of non-dense patterns are extracted alongside them (the rare item problem~\cite{Liu1999}).
Deep learning-based methods contribute to scalability and rule-count reduction; however, their support evaluation framework is identical to that of FIM, and approximate results based on latent representations do not guarantee exhaustive and exact extraction.

\subsection{Pattern Mining Considering Temporal Information}
\label{sec:rel_temporal}

Several methods have been proposed that account for changes in frequency along the time axis rather than evaluating frequency over the entire dataset.
TP-Mine~\cite{Wan2009} handles transitional patterns, detecting time points at which frequency changes dramatically; however, its output is change points rather than intervals.
SPF~\cite{Chang2002} computes frequency normalized by the exhibition period of each item, providing comparable support across the same item.
PPM~\cite{Lee2003} is a temporally correlated rule mining method that evaluates support for each interval partitioned at a fixed calendar granularity; while useful for capturing temporal changes in patterns, intervals are fixed to calendar granularity.

Burst detection~\cite{Kleinberg2003,Yang2021} identifies intervals in which the occurrence frequency of a single event surges; although it shares with dense patterns the motivation of detecting frequency concentration within specific intervals, it differs from our work in that it does not involve combinatorial search over itemsets.

In the field of periodic pattern mining, Periodic-Frequent Patterns (PFPs)~\cite{Tanbeer2009,Kiran2016,FournierViger2021b,Amphawan2009,Chen2023} have been proposed.
These methods output itemsets satisfying both $\mathit{Sup}(X) \geq \mathit{minSup}$ and $\mathit{Per}(X) \leq \mathit{maxPer}$, where the maximum period $\mathit{Per}(X) = \max_i p_i$ and $p_1, p_2, \ldots$ is the sequence of occurrence gaps of itemset $X$.
Since $\mathit{maxPer}$ acts as an upper bound on all occurrence gaps, smaller values require stricter periodicity; a single large gap causes the entire itemset to be rejected, making these methods susceptible to noise.
Partial Periodic-Frequent Patterns (PPFPs)~\cite{Kiran2017} relax this constraint and define qualifying itemsets as those with $\mathit{PR}(X) = |\{p_i \leq \mathit{maxPer}\}| / (|\mathit{Sup}(X)|+1) \geq \mathit{minPR}$.
The parameter $\mathit{minPR}$ controls the required fraction of occurrence gaps that lie within $\mathit{maxPer}$; lowering it allows detection of patterns with irregular periodicity.
However, both methods output only patterns and do not explicitly output dense intervals.

Episode mining~\cite{Radhakrishna2015} handles sequential patterns that involve temporal relationships among multiple events, capturing periodicity based on event occurrence frequency, but differs from our problem setting of simultaneously outputting itemsets and their dense intervals.
Emerging pattern mining~\cite{Dong1999} discovers patterns whose support changes substantially before and after a shift in data distribution, but does not produce dense interval output.

\subsection{Methods That Output Patterns and Their Dense Intervals}
\label{sec:rel_interval}

For dense pattern extraction, it is important to output not only the patterns but also the intervals in which they are dense.
Existing methods that extract intervals can be broadly categorized into those that directly compute local frequency and those based on occurrence-gap constraints.

As a direct local frequency approach, GLFMiner~\cite{Yin2014} automatically extracts locally frequent patterns and their intervals, but the intervals are fixed to integer multiples of a predefined temporal granularity, precluding data-driven flexible interval detection.
SIM~\cite{Saleh2008,Saleh2011} outputs the optimal single period for each itemset that satisfies a frequency condition.
Although its motivation is similar to dense patterns, it differs from our problem setting of comprehensively enumerating dense intervals in that it outputs a single optimal period per itemset and imposes no minimum interval length constraint.

As occurrence-gap-based approaches, LPFIM~\cite{Mahanta2005} extracts locally and periodically frequent itemsets with their intervals.
For an itemset $X$, a new interval is started when the gap between occurrences reaches $\mathit{minthd1}$; intervals satisfying local support $\mathit{Sup}_{[t_1,t_2]}(X) / \mathit{tc}_{[t_1,t_2]} \geq \sigma$ and interval length $\geq \mathit{minthd2}$ are output (where $t_c$ is the number of transactions in the interval).
The parameter $\mathit{minthd1}$ determines the interval segmentation granularity, and $\mathit{minthd2}$ removes intervals that are too short.

LPPM~\cite{FournierViger2021} achieves noise-robust interval detection by introducing cumulative spillover against the gap constraint.
The spillover at occurrence timestamp $ts_i$ is $\mathit{surPer}(ts_i) = \mathit{interval}(ts_i) - \mathit{maxPer}$, accumulated as $\mathit{soPer}(ts_i) = \max(0,\; \mathit{soPer}(ts_{i-1}) + \mathit{surPer}(ts_i))$; the interval is terminated when this exceeds $\mathit{maxSoPer}$.
Intervals of length $\geq \mathit{minDur}$ are output together with patterns.
The parameter $\mathit{maxPer}$ controls the upper bound on acceptable occurrence gaps, $\mathit{maxSoPer}$ controls the total permissible excess, and $\mathit{minDur}$ controls the minimum interval length.

RPM~\cite{Kiran2015} partitions intervals where the occurrence gap is at most $\mathit{maxPer}$ into periodic intervals, and outputs itemsets for which the number of periodic intervals containing at least $\mathit{minPS}$ consecutive occurrences is at least $\mathit{minRec}$.
The parameter $\mathit{maxPer}$ controls the upper bound on gaps considered periodic, $\mathit{minPS}$ controls the minimum number of consecutive occurrences within one interval, and $\mathit{minRec}$ controls the minimum number of periodic intervals.

These methods are the most closely related prior work to ours in that they output both patterns and their associated intervals.
Table~\ref{tab:pattern_def} compares the pattern definitions of each method with the dense pattern definition of this work.
In contrast to Apriori-window, which directly evaluates the occurrence count within a window, existing methods indirectly estimate density through parameters such as gap constraints, cumulative spillover, and periodicity thresholds.

\begin{table*}[t]
\centering
\caption{Comparison of pattern definitions across methods.}
\label{tab:pattern_def}
\footnotesize
\setlength{\tabcolsep}{4pt}
\renewcommand{\arraystretch}{1.25}
\begin{tabularx}{\textwidth}{@{}p{3.2cm} X p{4.4cm} p{2.0cm}@{}}
\toprule
Pattern Name & Pattern Condition & Dense Interval Detection & Output \\
\midrule
Periodic-Frequent Pattern &
  \textbullet\ $\mathit{Sup}(X) \geq \mathit{minSup}$: lower bound on total occurrence count\newline
  \textbullet\ $\mathit{Per}(X) = \max_i p_i \leq \mathit{maxPer}$: upper bound on the maximum gap between consecutive occurrences &
  Does not output intervals &
  $P$ only \\
Partial Periodic-Frequent Pattern &
  \textbullet\ $\mathit{Sup}(X) \geq \mathit{minSup}$: lower bound on total occurrence count\newline
  \textbullet\ $\mathit{PR}(X) \geq \mathit{minPR}$: lower bound on the fraction of occurrence gaps that are at most $\mathit{maxPer}$ &
  Does not output intervals &
  $P$ only \\
\midrule
Locally and Periodically Frequent Itemset &
  \textbullet\ $\mathit{Sup}_{[t_1,t_2]}(X)/t_c \geq \sigma_\ell$: lower bound on local support rate within $[t_1,t_2]$ ($t_c$: interval length)\newline
  \textbullet\ $t_c \geq \mathit{minthd2}$: minimum interval length &
  Partitions an interval when the occurrence gap reaches $\mathit{minthd1}$ &
  $P$ + interval list \\
Local Periodic Pattern &
  \textbullet\ $\mathit{spillover}_{[t_1,t_2]} \leq \mathit{maxSoPer}$: upper bound on the cumulative surplus of occurrence gaps exceeding $\mathit{maxPer}$\newline
  \textbullet\ $t_c \geq \mathit{minDur}$: minimum interval length &
  Terminates an interval when the cumulative surplus exceeds $\mathit{maxSoPer}$ &
  $P$ + interval list \\
Recurring Pattern &
  \textbullet\ $\mathit{rec}(X) \geq \mathit{minRec}$: lower bound on the number of detected recurring intervals\newline
  \textbullet\ $\mathit{ps}(X) \geq \mathit{minPS}$: lower bound on occurrence density within each recurring interval &
  Groups consecutive occurrences with gaps at most $\mathit{maxPer}$ into a single interval &
  $P$ + interval list \\
\midrule
Dense Pattern (Ours) &
  \textbullet\ $\sup(P,l,W) \geq \sigma$: lower bound on direct occurrence count in window $[l,l{+}W]$ &
  Determines whether the frequency condition is satisfied within a sliding window &
  $P$ + interval list \\
\bottomrule
\end{tabularx}
\renewcommand{\arraystretch}{1.0}
\end{table*}

\subsection{Evaluation of Local Frequency via Sliding Window}
\label{sec:rel_window}

Sliding windows~\cite{Lee2001,Leung2006} have been employed in stream mining~\cite{FournierViger2016spmf} and change point detection in time series~\cite{Aminikhanghahi2017}.
In stream mining, windows serve as a framework for limiting the temporal scope of data analysis as new data arrive online; in change point detection, they are used as a means of detecting local statistical changes in data.

The essence of dense patterns is that a given itemset occurs locally at high frequency within a specific consecutive interval.
In this work, we introduce the evaluation of local frequency within consecutive intervals using a sliding window.
By directly computing the occurrence frequency within the window as it moves, we can assess local density without relying on occurrence-gap constraints, thereby avoiding the parameter sensitivity problem in principle.
This idea is similar in spirit to related work on change point detection.
However, within our survey, we could not identify any prior work whose primary objective is the comprehensive extraction of itemsets and their dense intervals by combining direct local frequency evaluation via a sliding window with the anti-monotonicity of Apriori.
This work is an attempt to formulate and implement this direction as an algorithm.

\section{Problem Definition}
\label{sec:problem}

\subsection{Basic Definitions}

\begin{definition}[Timestamped Transaction Database]
A timestamped transaction database over an item universe $\mathcal{I}$ is a finite set
$D = \{(T_t, t) \mid t \in \mathbb{Z}_{\geq 0}\}$,
where each $T_t \subseteq \mathcal{I}$ is the transaction corresponding to timestamp $t$.
\end{definition}

\noindent\textit{Example.}
Throughout the following discussion, we use the database $D$ shown in Table~\ref{tab:example_db} ($\mathcal{I}=\{a,b,c\}$, $W=10$, $\sigma=3$).

\begin{table}[t]
\centering
\caption{Transaction database $D$ used in running examples.}
\label{tab:example_db}
\small
\begin{tabular}{rl}
\toprule
Timestamp $t$ & Transaction $T_t$ \\
\midrule
1  & $\{a, b\}$ \\
3  & $\{a, b, c\}$ \\
5  & $\{b, c\}$ \\
7  & $\{a, b, c\}$ \\
9  & $\{a, b\}$ \\
20 & $\{a, b, c\}$ \\
22 & $\{b, c\}$ \\
25 & $\{a, b, c\}$ \\
\bottomrule
\end{tabular}
\end{table}

\begin{definition}[Occurrence Timestamp Sequence]
The occurrence timestamp sequence of itemset $P \subseteq \mathcal{I}$ in $D$ is defined as
$\mathit{occ}(P, D) = \{t \mid P \subseteq T_t,\, (T_t, t) \in D\}$.
When $D$ is clear from context, we abbreviate this as $\mathit{occ}(P)$.
\end{definition}

\noindent\textit{Example.}
The occurrence timestamp sequences for each itemset in Table~\ref{tab:example_db} are as follows:
\begin{align*}
  \mathit{occ}(\{a\})     &= \{1, 3, 7, 9, 20, 25\} \\
  \mathit{occ}(\{b\})     &= \{1, 3, 5, 7, 9, 20, 22, 25\} \\
  \mathit{occ}(\{c\})     &= \{3, 5, 7, 20, 22, 25\} \\
  \mathit{occ}(\{a,b\})   &= \{1, 3, 7, 9, 20, 25\} \\
  \mathit{occ}(\{a,c\})   &= \{3, 7, 20, 25\} \\
  \mathit{occ}(\{b,c\})   &= \{3, 5, 7, 20, 22, 25\} \\
  \mathit{occ}(\{a,b,c\}) &= \{3, 7, 20, 25\}
\end{align*}

\begin{definition}[Interval Support]
The interval support of itemset $P$ in window $[l,\, l+W]$ (with window size $W > 0$) is defined as
\[
  \mathit{sup}(P,\, l,\, W) = |\{t \in \mathit{occ}(P) \mid l \leq t \leq l + W\}|.
\]
For minimum support $\sigma \in \mathbb{Z}_{>0}$, a window $[l, l+W]$ satisfying $\mathit{sup}(P, l, W) \geq \sigma$ is called a \emph{dense window} of $P$.
\end{definition}

\noindent\textit{Example.}
For $P=\{a,b\}$, $l=1$, $W=10$:
\[
  \mathit{sup}(\{a,b\},1,10)=|\{1,3,7,9\}|=4\geq3=\sigma,
\]
so window $[1,11]$ is a dense window of $\{a,b\}$.
At $l=15$, however, $\mathit{sup}(\{a,b\},15,10)=|\{20,25\}|=2<3$, so $[15,25]$ is not a dense window.

\begin{definition}[Dense Interval]
An interval $[s, e]$ (with $0 \leq s$, $e \leq T_{\max}$, and $e - s \geq W$) is a \emph{dense interval} of $P$ if
{\small\[
  \forall\, l \in \bigl\{ t \in \mathit{occ}(P) \mid s \leq t \leq e{-}W \bigr\}:\;
  \mathit{sup}(P, l, W) \geq \sigma
\]}
and the interval cannot be further extended. Here $T_{\max}$ is the maximum timestamp in dataset $D$.
\end{definition}

\noindent\textit{Example.}
For $P=\{a,b\}$ with $\mathit{occ}(\{a,b\})=\{1,3,7,9,20,25\}$, the condition $\mathit{sup}(\{a,b\},l,10)\geq3$ holds only for $l\in[-3,3]$.
Indeed, at $l=4$, $\mathit{sup}(\{a,b\},4,10) = |\{7,9\}|=2<3$.
Hence the right endpoint of the dense interval is $e=3+10=13$, and constraining the left endpoint to $\geq 0$, the dense interval of $\{a,b\}$ is $[0,13]$.

\begin{definition}[Dense Pattern]
An itemset $P$ that has at least one dense interval is called a \emph{dense pattern}.
\end{definition}

\begin{definition}[Dense Pattern Mining Problem]
Given a timestamped transaction database $D$, a window size $W$, and a minimum support $\sigma$, enumerate all pairs $(P, \text{DI}(P))$ of dense patterns $P$ and their dense interval sets $\text{DI}(P)$.
\end{definition}

\noindent\textit{Example.}
For the database $D$ in Table~\ref{tab:example_db} with $W=10$, $\sigma=3$, and $T_{\max}=25$, all dense patterns and their dense intervals are:
\begin{align*}
  \{a\}   &:\ [0,13] \\
  \{b\}   &:\ [0,15],\ [15,25] \\
  \{c\}   &:\ [0,13],\ [15,25] \\
  \{a,b\} &:\ [0,13] \\
  \{b,c\} &:\ [0,13],\ [15,25]
\end{align*}

\section{Proposed Method: Apriori-window}
\label{sec:method}

\subsection{Method Overview}

Apriori-window is an exact algorithm that applies the anti-monotonicity of Apriori~\cite{Agrawal1994} to sliding-window local frequency evaluation.
The input is a timestamped transaction database $D$, a window size $W$, and a minimum support $\sigma$; the output is an enumeration of all dense patterns $P$ together with their dense interval sets $\text{DI}(P)$.

The algorithm begins at itemset size $k=1$, repeatedly generating size-$(k+1)$ candidates from the size-$k$ dense pattern set $L_k$ and computing dense intervals for each candidate.
Search space reduction employs two mechanisms: (1)~candidate pattern pruning via anti-monotonicity of itemsets, and (2)~candidate interval restriction via the intersection of dense intervals of individual items.
Algorithm~\ref{alg:main} presents the Apriori-window algorithm.

\begin{algorithm}[ht]
  \caption{Apriori-window}
  \label{alg:main}
  \begin{algorithmic}[1]
  \REQUIRE item\_ts: item-to-occ map; $W$: window size; $\sigma$: min support
  \ENSURE $F$: map from pattern to dense intervals
  \STATE $F, L_1, S \leftarrow \emptyset, \emptyset$, empty map
  \FOR{$x$ in sorted keys(item\_ts)}
    \STATE DI$(x) \leftarrow$ Scan(occ$(x), W, \sigma, \text{unr})$
    \IF{DI$(x) \neq \emptyset$}
      \STATE $F[\{x\}] \leftarrow$ DI$(x)$; $L_1 \leftarrow L_1 \cup \{\{x\}\}$; $S[x] \leftarrow$ DI$(x)$
    \ENDIF
  \ENDFOR
  \STATE $k \leftarrow 2$; $L_p \leftarrow L_1$
  \WHILE{$L_p \neq \emptyset$}
    \STATE $C_k \leftarrow$ AprioriJoin$(L_p)$ on first $k{-}2$ items
    \STATE $C_k \leftarrow \{X \in C_k \mid$ all $(k{-}1)$-subsets in $L_p\}$
    \STATE $L_k \leftarrow \emptyset$
    \FOR{$X \in C_k$}
      \STATE $\text{occ}(X) \leftarrow \bigcap_{x \in X} \text{occ}(x)$
      \IF{$\text{occ}(X) = \emptyset$}
        \STATE continue
      \ENDIF
      \STATE $R_X \leftarrow$ cand$(X)$ using $S[x]$ for $x \in X$
      \IF{$R_X = \emptyset$}
        \STATE continue
      \ENDIF
      \STATE DI$(X) \leftarrow$ Scan$(\text{occ}(X), W, \sigma, R_X)$
      \IF{DI$(X) \neq \emptyset$}
        \STATE $F[X] \leftarrow$ DI$(X)$; $L_k \leftarrow L_k \cup \{X\}$
      \ENDIF
    \ENDFOR
    \STATE $L_p \leftarrow L_k$; $k \leftarrow k+1$
  \ENDWHILE
  \RETURN $F$
  \end{algorithmic}
\end{algorithm}

\subsection{Anti-monotonicity}

\begin{proposition}[Anti-monotonicity of Interval Support]
If $P' \supseteq P$, then $\mathit{sup}(P', l, W) \leq \mathit{sup}(P, l, W)$ for any $l$ and $W$.
\end{proposition}

\begin{proof}
Since $P' \supseteq P$, any occurrence of $P'$ must contain all items of $P$, so $\mathit{occ}(P', D) \subseteq \mathit{occ}(P, D)$.
Consequently, the occurrence count within any window also satisfies $\mathit{sup}(P', l, W) \leq \mathit{sup}(P, l, W)$.
\end{proof}

Proposition~1 establishes that the same candidate generation and pruning strategy as Apriori applies to the dense pattern mining problem.

\begin{proposition}[Containment of Dense Intervals]
For any $x \in P$, every dense interval of $P$ is contained within some dense interval of $x$.
\end{proposition}

\begin{proof}
For a dense interval $[s, e]$ of $P$, there exists $l \in [s, e-W]$ such that $\mathit{sup}(P, l, W) \geq \sigma$.
By Proposition~1 (anti-monotonicity), for any $x \in P$, $\mathit{sup}(x, l, W) \geq \mathit{sup}(P, l, W) \geq \sigma$, so $[l, l+W]$ is a dense window of $x$.
Hence $l$ is contained within some dense interval of $x$.
Since the left endpoint of $[s, e]$ is determined by $l$, the interval $[s, e]$ is contained within a dense interval of $x$.
\end{proof}

From Proposition~2, the dense intervals of a pattern $P$ of size $k \geq 2$ cannot exist outside the intersection of the dense intervals of all individual items composing $P$ (i.e., $\bigcap_{x \in P} \text{DI}(x)$).
Using this property, we restrict the search range for dense intervals to the candidate interval set
\[
  \mathit{cand}(P) = \left\{(s, e) \in \bigcap_{x \in P} \text{DI}(x) \;\middle|\; e - s \geq W \right\}.
\]
If $\mathit{cand}(P)$ is empty, then $P$ cannot be a dense pattern, and the search is terminated.

\subsection{Dense Interval Detection via Sliding Window}
\label{sec:scan}

Algorithm~\ref{alg:scan} presents the dense interval detection procedure.
In sliding-window interval detection, the number of scan operations affects computational efficiency.
To reduce the number of scans, we introduce stride-skip.

\begin{algorithm}[!h]
\caption{Scan}
\label{alg:scan}
\begin{algorithmic}[1]
\REQUIRE occ: sorted occurrence positions; $W$: window size; $\sigma$: min support; $R$: scan ranges
\ENSURE DI: dense interval list
\STATE DI, recorded $\leftarrow \emptyset$
\FOR{$(c_s, c_e) \in R$}
  \STATE $j \leftarrow$ LB(occ, $c_s$); $l \leftarrow \max(c_s, \max(0, \text{occ}[j+\sigma-1]-W))$
  \STATE in\_dense, $s$, $d \leftarrow$ false, $\bot$, $\bot$
  \WHILE{$l \leq c_e$}
    \IF{$l \in$ recorded block $[a, b]$}
      \STATE $l \leftarrow b+1$; continue
    \ENDIF
    \STATE cnt $\leftarrow$ occurrence in window $[l, l+W]$
    \IF{cnt $< \sigma$}
      \IF{in\_dense and $d+W-s \geq W$}
        \STATE Register$(s, d+W, \text{recorded})$
      \ENDIF
      \STATE in\_dense, $s$, $d \leftarrow$ false, $\bot$, $\bot$
      \STATE $l \leftarrow \max(l+1, \text{occ}[\text{UB}(occ, l+W)]-W)$
    \ELSE
      \STATE keep $\leftarrow \text{occ}[\text{LB}(occ, l) + \text{cnt} - \sigma]$
      \IF{not in\_dense}
        \STATE $s \leftarrow l$; $d \leftarrow$ keep; in\_dense $\leftarrow$ true
      \ELSE
        \STATE $d \leftarrow \max(d,$ keep$)$
      \ENDIF
      \IF{keep $+1 > c_e$}
        \STATE $d \leftarrow \min(d, c_e)$; break
      \ELSE
        \STATE $l \leftarrow$ keep $+1$
      \ENDIF
    \ENDIF
  \ENDWHILE
  \IF{in\_dense and $d+W-s \geq W$}
    \STATE Register$(s, d+W, \text{recorded})$
  \ENDIF
\ENDFOR
\RETURN DI
\end{algorithmic}
\end{algorithm}

\textbf{Stride-Skip:}
When the occurrence count within the window $[l, l+W]$ satisfies $\text{cnt} \geq \sigma$, let $t_{\text{keep}}$ be the position of the first of the last $\sigma$ occurrences; the window starting position $l$ is then advanced to $t_{\text{keep}}+1$.

\begin{proposition}[Safety of Stride-Skip]
Let $\text{cnt}$ be the occurrence count within the window $[l, l+W]$, and suppose $\text{cnt} \geq \sigma$.
Let $t_{\text{keep}}$ be the first of the last $\sigma$ occurrences in this window.
Then for any $l'$ satisfying $l \leq l' \leq t_{\text{keep}}$, the window $[l', l'+W]$ is also a dense window.
Therefore, advancing $l$ to $t_{\text{keep}} + 1$ does not miss any dense interval.
\end{proposition}

\begin{proof}
Let the occurrences within the window $[l, l+W]$ be ordered as $t_1 \leq t_2 \leq \cdots \leq t_{\text{cnt}}$ ($\text{cnt} \geq \sigma$).
Consider the last $\sigma$ occurrences $t_{\text{cnt}-\sigma+1}, t_{\text{cnt}-\sigma+2}, \ldots, t_{\text{cnt}}$.
For any $l'$ satisfying $l \leq l' \leq t_{\text{cnt}-\sigma+1}$, since $l' \leq t_{\text{cnt}-\sigma+1}$ and $t_{\text{cnt}} \leq l + W \leq l' + W$ (by $l' \geq l$), all $\sigma$ occurrences are contained in $[l', l'+W]$.
Hence $[l', l'+W]$ is a dense window.
\end{proof}

Proposition~3 guarantees the safety of stride-skip.
As the occurrence count $\text{cnt}$ within the window increases, $\text{cnt} - \sigma$ increases and the skip distance grows; therefore, greater computational efficiency is expected for patterns whose occurrences are more densely concentrated.

\textbf{Answer to RQ1:} Exploiting the anti-monotonicity of the pattern space and the containment relationship of dense intervals, it is possible to realize an exact algorithm by extending the Apriori framework. Stride-skip further improves efficiency while preserving the exactness of the interval search.

\section{Experiments}
\label{sec:experiments}

\subsection{Experiment A: Analysis of Existing Methods}
\label{sec:exp_comparison}

The purpose of this experiment is not to directly compare the proposed method against existing methods, but rather to investigate to what extent existing occurrence-gap-based methods can approximate the ``window-based local frequency dense patterns and dense intervals'' as defined in this work.
To this end, we use the output of Apriori-window, which provides the exact solution under our definition, as the reference, and evaluate how well existing methods can approximate this definition.
We use three benchmark datasets---Retail, OnlineRetail, and ChicagoCrime---publicly available through SPMF~\cite{FournierViger2016spmf}.
Dataset statistics are shown in Table~\ref{tab:datasets}.

\begin{table}[h]
\centering
\caption{Experimental datasets.}
\label{tab:datasets}
\footnotesize
\begin{tabular}{lrrr}
\toprule
Dataset & \#Trans & \#Items & Avg.\ length \\
\midrule
Retail       & 88,162    & 14,089 & 10.3 \\
OnlineRetail & 540,455   &  2,147 &  4.4 \\
ChicagoCrime & 2,662,309 &     34 &  1.8 \\
\bottomrule
\end{tabular}
\end{table}

We compare against the following methods:
\begin{itemize}
  \item FIM: Frequent Itemset Mining~\cite{Agrawal1994}
  \item PFPM: Periodic-Frequent Pattern Mining~\cite{Tanbeer2009}
  \item PPFPM: Partial Periodic-Frequent Pattern Mining~\cite{Kiran2017}
  \item LPFIM: Locally and Periodically Frequent Itemset Mining~\cite{Mahanta2005}
  \item LPPM: Local Periodic Pattern Mining~\cite{FournierViger2021}
  \item RPM: Recurring Pattern Mining~\cite{Kiran2015}
  \item Span: A reference method that outputs the interval from the first to the last occurrence timestamp of each pattern as a single interval. Its purpose is to verify the locality of the ground-truth dense periods in each dataset.
\end{itemize}

The parameters $(W, \sigma)$ of Apriori-window are set to $\{(250,25),(500,50),(750,75),(1000,100)\}$ for Retail and ChicagoCrime, and to $\{(100,10),(250,25),(500,50)\}$ for OnlineRetail.
Evaluation is restricted to patterns of length $\geq 2$; the number of ground-truth dense patterns for each setting is shown in Table~\ref{tab:gt_patterns}.

\begin{table}[t]
  \centering
  \caption{Number of ground-truth dense patterns (length $\geq 2$) in Experiment~A. Header values denote $W/\sigma$.}
  \label{tab:gt_patterns}
  \setlength{\tabcolsep}{3pt}
  \begin{tabular}{lrrrr}
  \toprule
  Dataset & $250/25$ & $500/50$ & $750/75$ & $1000/100$ \\
  \midrule
  Retail & 171 & 110 & 89 & 81 \\
  Chicago & 63 & 42 & 28 & 23 \\
  \bottomrule
  \end{tabular}

  \vspace{0.3cm}

  \begin{tabular}{lrrr}
  \toprule
  Dataset & $100/10$ & $250/25$ & $500/50$ \\
  \midrule
  OnlineRetail & 355 & 64 & 11 \\
  \bottomrule
  \end{tabular}
\end{table}

The extraction accuracy of each method is evaluated using three metrics: F1 score, mean Jaccard, and mean Temporal Precision.
Methods that do not output intervals (FIM, PFPM, PPFPM) are evaluated using F1 score only.
Let $\mathcal{P}^*$ denote the set of patterns extracted by Apriori-window and $\text{DI}^*(P)$ their dense intervals.
For each pattern $P$, let $G(P) = \bigcup_{[s,e]\in\text{DI}^*(P)}[s,e]$ and $\hat{G}(P) = \bigcup_{[s,e]\in\widehat{\text{DI}}(P)}[s,e]$ denote the temporal coverage of the ground-truth and predicted intervals, respectively (where $|\cdot|$ denotes duration), and let $\hat{\mathcal{P}}$ denote the set of predicted patterns for each method.

The F1 score is the harmonic mean based on the agreement between ground-truth and predicted pattern sets:
\[
  \text{F1} = \frac{2\,|\hat{\mathcal{P}} \cap \mathcal{P}^*|}{|\hat{\mathcal{P}}| + |\mathcal{P}^*|}
\]

Mean Jaccard is the average over all ground-truth patterns of the Jaccard coefficient computed from the temporal overlap between ground-truth and predicted intervals:
\[
  \text{mean Jaccard} = \frac{1}{|\mathcal{P}^*|} \sum_{P \in \mathcal{P}^*}
    \frac{|G(P) \cap \hat{G}(P)|}{|G(P) \cup \hat{G}(P)|}
\]

Mean Temporal Precision is computed only for correctly identified patterns ($P \in \hat{\mathcal{P}} \cap \mathcal{P}^*$) and is the average fraction of the predicted intervals that overlap with the ground-truth intervals:
\[
  \text{mean TP} = \frac{1}{|\hat{\mathcal{P}} \cap \mathcal{P}^*|}
    \sum_{P \in \hat{\mathcal{P}} \cap \mathcal{P}^*}
    \frac{|G(P) \cap \hat{G}(P)|}{|\hat{G}(P)|}
\]

The extraction accuracy of existing methods varies substantially depending on parameter settings.
Therefore, for each method, parameters are selected by grid search to maximize the F1 score; ties are broken by prioritizing the setting with the higher Jaccard coefficient.
The parameter search spaces for each method are shown in Table~\ref{tab:all_params}, and the best-F1 parameters for each method are listed in Appendix~\ref{sec:best_params}.

\begin{table}[t]
\centering
\caption{Parameter search space for each baseline method.
$\sigma$: minimum support in the experimental setting ($=W/10$).
$N$: upper bound for PFPM maximum period (Retail: 900, OnlineRetail: 3{,}500, Chicago: 10{,}000).
$N'$: upper bound for PPFPM maximum period (Retail: 50, OnlineRetail: 500, Chicago: 900).}
\label{tab:all_params}
\footnotesize
\setlength{\tabcolsep}{3pt}
\begin{tabularx}{\columnwidth}{@{}p{1.0cm}p{2.4cm}Xc@{}}
\toprule
Method & Parameter & Value / Search Range & \\
\midrule
FIM & $\mathit{minSup}$ & $\sigma$ & fixed \\
\midrule
\multirow{2}{*}{PFPM}
 & $\mathit{maxPer}$ & $\{5,10,\ldots,N\}$ (step 5) & searched \\
 & $\mathit{minSup}$ & $\sigma$ & fixed \\
\midrule
\multirow{3}{*}{PPFPM}
 & $\mathit{maxPer}$ & $\{5,10,\ldots,N'\}$ (step 5) & searched \\
 & $\mathit{minPR}$  & $\{0.05,0.10,\ldots,0.95\}$ & searched \\
 & $\mathit{minSup}$ & $\sigma$ & fixed \\
\midrule
\multirow{4}{*}{LPFIM}
 & $\sigma_\ell$ & $\{10,15,20,25,30\}$\% & searched \\
 & $\mathit{minthd1}$ & $\{10,15,\ldots,100\}$ (step 5) & searched \\
 & $\mathit{minthd2}$ & $\{10,20,\ldots,1000\}$ (step 10) & searched \\
\midrule
\multirow{3}{*}{LPPM}
 & $\mathit{maxPer}$ & $\{5,10,\ldots,50\}$ (step 5) & searched \\
 & $\mathit{minDur}$ & $\{5,10,\ldots,1000\}$ (step 5) & searched \\
 & $\mathit{maxSoPer}$ & $\{5,10,\ldots,500\}$ (step 5) & searched \\
\midrule
\multirow{3}{*}{RPM}
 & $\mathit{maxPer}$ & $\{5,10,\ldots,250\}$ (step 5) & searched \\
 & $\mathit{minPS}$ & $\sigma$ & fixed \\
 & $\mathit{minRec}$ & 1 & fixed \\
\bottomrule
\end{tabularx}
\end{table}

LPPM is implemented using the Java code publicly available in SPMF~\cite{FournierViger2016spmf}; all other existing methods and Apriori-window are implemented in Rust.

Table~\ref{tab:combined} presents the F1 score, mean Jaccard, and mean Temporal Precision at the best parameters for all methods across three datasets and all window sizes.
The F1 score of FIM remained at most 0.035.
Since FIM extracts patterns based on dataset-wide support, it retrieves a large number of non-dense patterns that occur at similar frequency to dense patterns but are spread throughout the dataset; consequently, recall is 1.0 while precision is extremely low.
PFPM and PPFPM show improved F1 compared to FIM through the introduction of gap constraints (PPFPM Retail: 0.782--0.926, Chicago: 0.740--0.818), suggesting that occurrence-gap constraints can contribute to pattern identification accuracy.

LPFIM, LPPM, and RPM, which output intervals, achieve high F1 (0.866--0.989) in certain settings for Retail and Chicago; however, Jaccard shows considerable variation across datasets and window sizes, and is particularly low for OnlineRetail.
This is presumably because the ground-truth dense intervals include interval expansion based on the sliding window width $W$, which occurrence-gap-based methods cannot reproduce by definition.
Based on mean Temporal Precision, we confirmed that interval mis-detection occurs even in settings where patterns are correctly identified.
Across all methods and all datasets, mean Temporal Precision exceeds that of the Span reference method, suggesting that occurrence-gap-based interval extraction provides a certain degree of filtering effectiveness.
However, in datasets where the Span score is low, cases were observed in which TP decreases even at settings with high F1 (e.g., RPM Chicago $W=1000$: F1 = 0.920, TP = 0.591; LPPM OnlineRetail $W=500$: F1 = 0.957, TP = 0.358).
This suggests that relaxing parameter constraints to avoid missing dense patterns causes interval false positives, revealing a trade-off between pattern identification and accurate interval extraction.

\textbf{Answer to RQ2:} Existing methods show a degree of effectiveness in dense pattern identification; however, simultaneously achieving high accuracy in both pattern identification and dense interval detection may be difficult.

\begin{table*}[t]
\centering
\caption{Pattern identification accuracy and interval detection accuracy of baseline methods at best-F1 parameters.
Best parameters for each method are listed in Appendix~\ref{sec:best_params}.
Since Apriori-window is an exact algorithm for dense patterns, it achieves F1 = Jaccard = TP = 1.000.}
\label{tab:combined}
\scriptsize
\setlength{\tabcolsep}{1.3pt}
\resizebox{\textwidth}{!}{%
\begin{tabular}{llr|rrrrrr|rrrr|rrrr}
\toprule
 & & & \multicolumn{6}{c|}{F1 Score} & \multicolumn{4}{c|}{Jaccard} & \multicolumn{4}{c}{Temporal Precision} \\
\cmidrule(lr){4-9}\cmidrule(lr){10-13}\cmidrule(lr){14-17}
Dataset & $W$ & $\sigma$ & FIM & PFPM & PPFPM & LPFIM & LPPM & RPM & LPFIM & LPPM & RPM & Span & LPFIM & LPPM & RPM & Span \\
\midrule
\multirow{4}{*}{Retail}
 & 250  & 25  & .009 & .119 & .782 & .866 & .917 & .906 & .205 & .637 & .652 & .397 & .896 & .970 & .936 & .428 \\
 & 500  & 50  & .017 & .120 & .861 & .830 & .967 & .941 & .218 & .715 & .710 & .526 & .949 & .988 & .933 & .609 \\
 & 750  & 75  & .025 & .146 & .908 & .867 & .989 & .943 & .206 & .695 & .713 & .564 & .957 & .977 & .942 & .691 \\
 & 1000 & 100 & .035 & .159 & .926 & .868 & .981 & .929 & .201 & .649 & .652 & .540 & .962 & .954 & .945 & .704 \\
\midrule
\multirow{4}{*}{Chicago}
 & 250  & 25  & .007 & .615 & .740 & .897 & .969 & .928 & .196 & .502 & .451 & .035 & .752 & .761 & .749 & .035 \\
 & 500  & 50  & .007 & .627 & .743 & .886 & .952 & .963 & .337 & .514 & .424 & .042 & .817 & .853 & .716 & .042 \\
 & 750  & 75  & .006 & .604 & .776 & .945 & .929 & .963 & .303 & .150 & .330 & .058 & .743 & .754 & .796 & .058 \\
 & 1000 & 100 & .007 & .622 & .818 & .957 & .936 & .920 & .156 & .182 & .346 & .069 & .673 & .730 & .591 & .069 \\
\midrule
\multirow{3}{*}{OnlineRetail}
 & 100  & 10  & .013 & .312 & .770 & .732 & .774 & .938 & .040 & .065 & .216 & .012 & .728 & .762 & .896 & .012 \\
 & 250  & 25  & .003 & .230 & .623 & .814 & .832 & .699 & .027 & .097 & .311 & .008 & .432 & .637 & .741 & .008 \\
 & 500  & 50  & .001 & .500 & .909 & 1.000 & .957 & .909 & .135 & .239 & .163 & .013 & .687 & .358 & .244 & .013 \\
\bottomrule
\end{tabular}%
}
\end{table*}

\begin{figure}[t]
  \centering
  \includegraphics[width=\columnwidth]{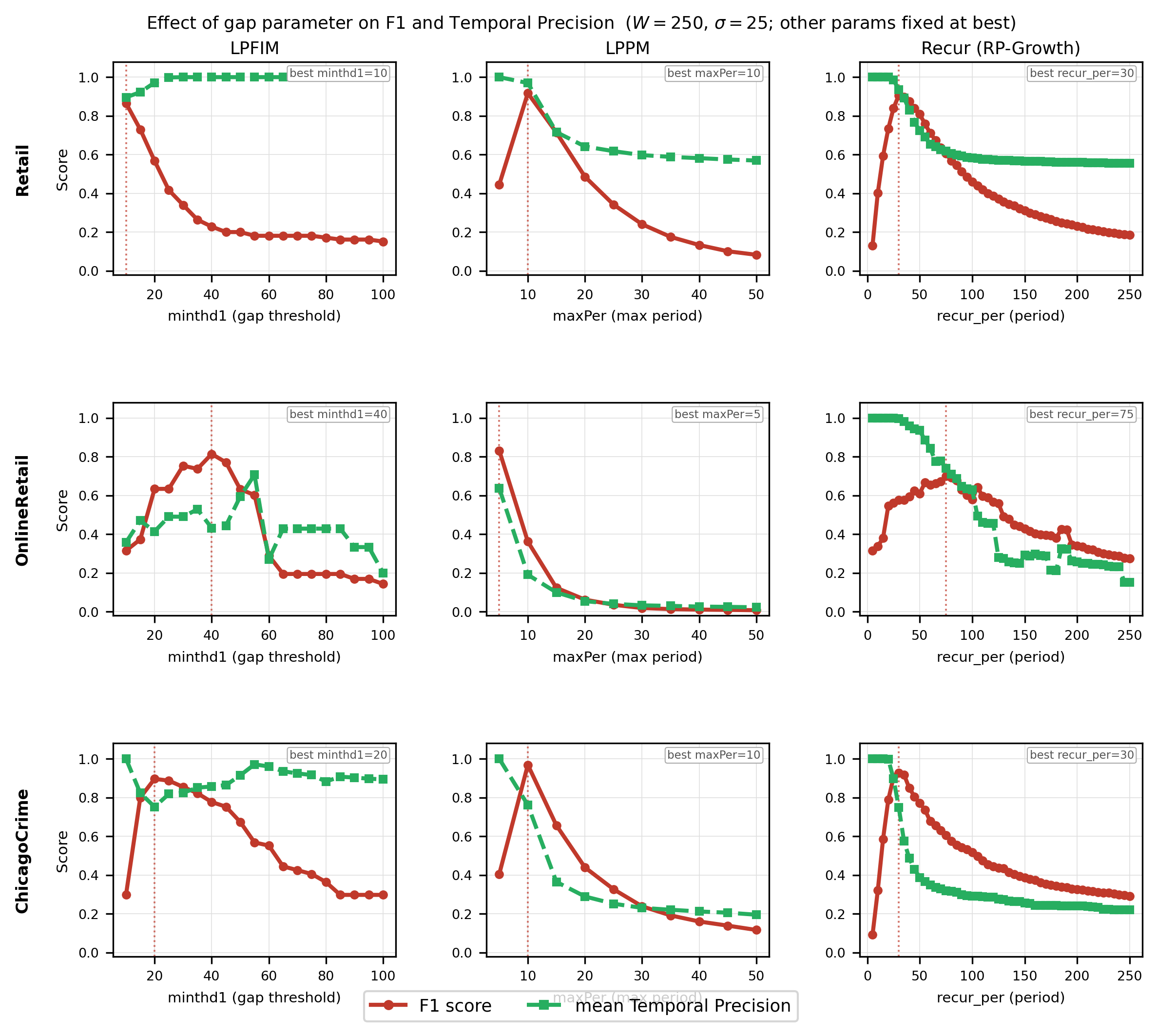}
  \caption{Change in F1 score (red, solid) and mean Temporal Precision (green, dashed) with respect to the gap constraint parameter (LPFIM: $\mathit{minthd1}$, LPPM: $\mathit{maxPer}$, RPM: recur\_per) for each method ($W=250$, $\sigma=25$; other parameters fixed at best-F1 settings). Vertical dotted lines indicate the gap value that maximizes F1.}
  \label{fig:gap_tp}
\end{figure}

\subsection{Experiment B: Scalability}

This experiment examines the growth trend of Apriori-window in execution time and memory usage with respect to data characteristics.
Using synthetic data, we run 45 configurations with varying parameters $T \in \{1.0, 1.25, 1.5, 1.75, 2.0\} \times 10^6$ transactions, $I \in \{10{,}000, 15{,}000, 20{,}000\}$ items, and mean basket length $B \in \{5, 10, 15\}$, each repeated 10 times, and measure the mean execution time and peak memory usage.
The synthetic data generation procedure is described in Appendix~\ref{sec:app_synthetic}.
Experiments were conducted on an 11th Gen Intel Core i7-1185G7 @ 3.00~GHz with 32~GB RAM.

Execution time and peak memory usage are shown in Fig.~\ref{fig:scale} and Fig.~\ref{fig:memory}, respectively.
Both metrics grow approximately linearly with the number of transactions, number of items, and basket length; even under the maximum setting ($T=2\times10^6$, $I=20{,}000$, $B=15$), the mean execution time was 8.9 seconds and memory usage was 333~MiB.
The number of extracted patterns ranged from 1,545 to 1,550 across all settings, confirming that the variation in dense pattern count due to parameter changes is small and that the main drivers of execution time and memory variation are data size and basket length.
Both figures show increased execution time and memory usage as mean basket length increases (from 5 to 15), presumably because a larger basket length increases both the number of occurrence timestamps for each item and the number of candidate patterns.
By contrast, the effect of increasing the number of items $I$ (10k to 20k) is relatively small compared to basket length, presumably because while increasing items expands the length-1 candidate set, anti-monotonicity-based pruning restricts the search for longer patterns to those that include the embedded patterns.

\textbf{Answer to RQ3:} Execution time grows approximately linearly with the number of transactions, items, and basket length; even under the maximum setting ($T=2\times10^6$, $I=20\mathrm{k}$, $B=15$), processing was completed within 8.9 seconds and 333~MiB of memory.

\begin{figure}[ht]
  \centering
  \begin{minipage}[t]{0.48\textwidth}
    \centering
    \includegraphics[width=\linewidth]{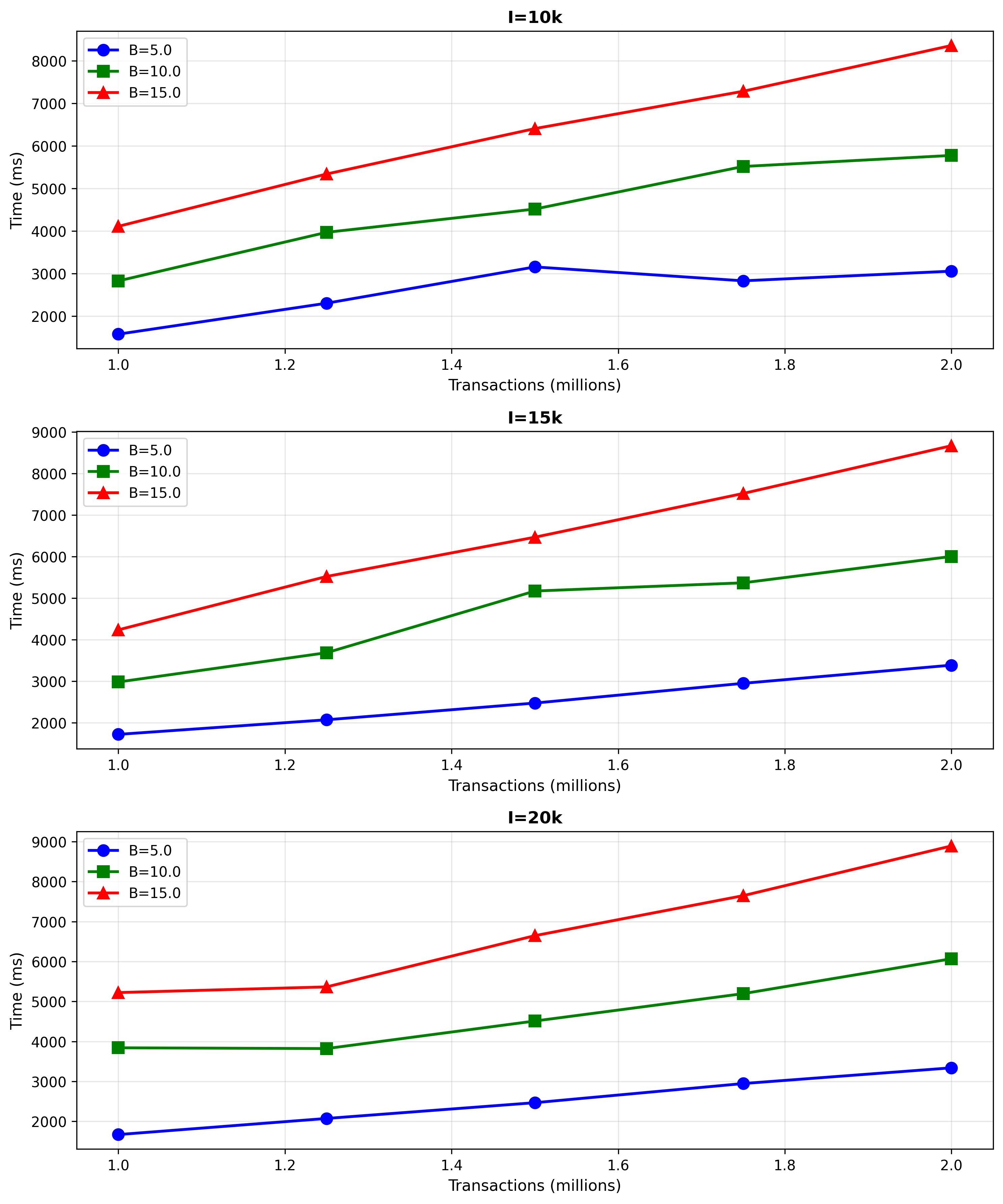}
    \caption{Scalability experiment: mean execution time ($T \in [1{\times}10^6,\,2{\times}10^6]$, $I \in \{10\mathrm{k},15\mathrm{k},20\mathrm{k}\}$, $B \in \{5,10,15\}$, averaged over 10 trials).}
    \label{fig:scale}
  \end{minipage}
  \hfill
  \begin{minipage}[t]{0.48\textwidth}
    \centering
    \includegraphics[width=\linewidth]{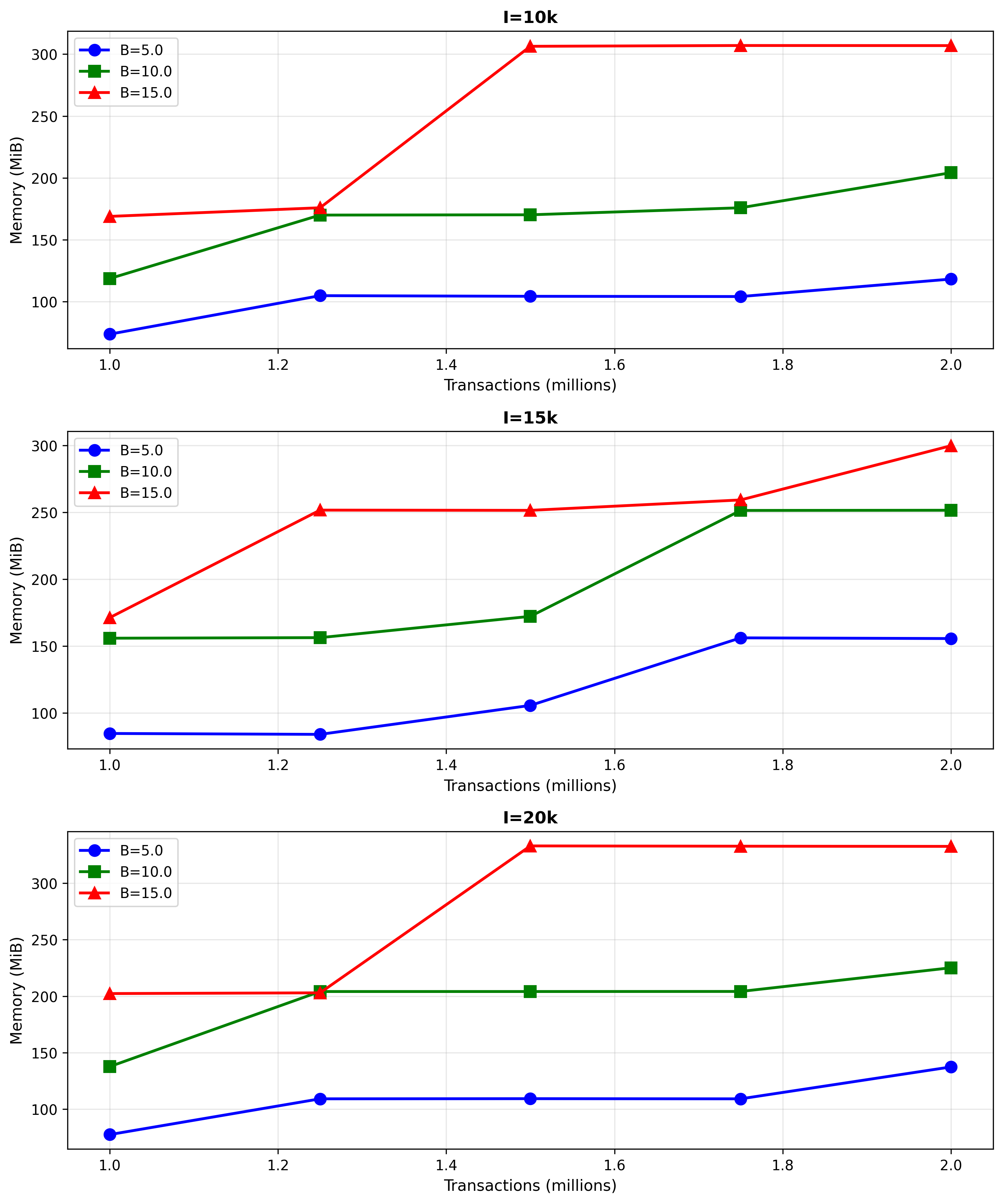}
    \caption{Scalability experiment: peak memory usage (same conditions, averaged over 10 trials). Line styles correspond to $B$ settings; colors correspond to $I$ settings.}
    \label{fig:memory}
  \end{minipage}
\end{figure}


\subsection{Experiment C: Ablation Study}
\label{sec:exp_ablation}

We evaluate the effect of the two search space reduction techniques employed by the proposed method.
The models used in the ablation study are defined as follows:
\begin{itemize}
\item \textbf{Baseline}: A basic implementation with no optimizations applied. Candidate search is performed at all window positions.
\item \textbf{Intersect}: Applies candidate reduction via the intersection of dense intervals of length-1 patterns, but still performs window sliding at all positions.
\item \textbf{Apriori-window}: In addition to the intersection-based reduction, applies the stride-skip optimization. Windows are skipped based on the number of surplus occurrences, reducing unnecessary scan positions.
\end{itemize}

\begin{table}[t]
  \centering
  \caption{Ablation study: comparison of execution times (ms, averaged over 10 trials) for different combinations of optimization components.}
  \label{tab:ablation}
  \footnotesize
  \setlength{\tabcolsep}{3pt}
  \begin{tabular}{llrrrr}
  \toprule
  Dataset & $W$ & $\sigma$ & Baseline & Intersect & Apriori-window \\
  \midrule
  \multirow{4}{*}{Retail}
   & 250  & 25  & 3,313 & 3,111 &   164 \\
   & 500  & 50  & 3,181 & 3,241 &   324 \\
   & 750  & 75  & 3,744 & 3,145 &   184 \\
   & 1000 & 100 & 3,190 & 3,106 &   140 \\
  \midrule
  \multirow{4}{*}{Chicago}
   & 250  & 25  & 6,587 & 3,575 & 2,036 \\
   & 500  & 50  & 5,202 & 3,023 & 1,303 \\
   & 750  & 75  & 4,626 & 2,953 & 1,116 \\
   & 1000 & 100 & 4,269 & 2,806 & 1,077 \\
  \midrule
  \multirow{3}{*}{OnlineRetail}
   & 100  & 10  & 14,620 & 6,020 &   813 \\
   & 250  & 25  & 7,320  & 5,631 &   488 \\
   & 500  & 50  & 6,475  & 5,642 &   357 \\
  \bottomrule
  \end{tabular}
\end{table}

Table~\ref{tab:ablation} presents the ablation study results.
First, the Intersect model reduces execution time compared to Baseline by an approximately equivalent margin for Retail, 40--50\% for Chicago, and 10--60\% for OnlineRetail.
When the intersection computation substantially narrows the candidate interval space, the reduction in the candidate search space outweighs the overhead.
Second, Apriori-window achieves reductions of 94--96\% overall for Retail, 65--77\% for Chicago, and 88--94\% for OnlineRetail, demonstrating the pronounced effect of the stride-skip optimization.
The relative contribution of stride-skip tends to improve as the window size increases, presumably because a larger window leads to greater surplus and consequently a wider skip distance.

\subsection{Experiment D: Application to Real-World Data}
\label{sec:exp_qualitative}

This experiment applies the proposed method to real-world transactional data annotated with external events, and investigates whether dense patterns with different correspondences to external events can be extracted by varying the window size.

We use the Dunnhumby Complete Journey dataset.\footnote{https://www.kaggle.com/datasets/frtgnn/dunnhumby-the-complete-journey}
This dataset contains two years (711 days) of purchase histories from a U.S.\ retailer, consisting of 276,176 baskets and 308 product categories.
The dataset also includes records of 30 promotional campaigns (Type~A, Type~B, and Type~C, each lasting 32--161 days, with a mean duration of 46.6 days) and the products for which coupons were distributed in each campaign.
Data preprocessing is described in Appendix~\ref{sec:app_preprocess}.

We run Apriori-window under the following two conditions:
\begin{itemize}
  \item \textbf{$W=47$ days}: $W = 47{,}000$ ($47\ \text{days} \approx \text{mean promotional duration}$), $\sigma = 470$, $k_{\max} = 5$. Intended for detecting dense intervals corresponding to individual promotional periods.
  \item \textbf{$W=141$ days}: $W = 141{,}000$ ($141\ \text{days} \approx 3 \times \text{mean duration}$), $\sigma = 1{,}410$, $k_{\max} = 5$. Intended for extracting long-term patterns spanning multiple promotional periods.
\end{itemize}

As an indicator of the relationship between each dense pattern's intervals and the promotional periods, we introduce the overlap ratio.
For each pattern, we compute the overlap ratio $R_{\text{promo}}$ between its dense intervals and the promotional periods of the constituent products on a day-count basis:
\[
  R_{\text{promo}} = \frac{\sum_i \text{overlap}([\,s_i, e_i\,],\ P_{\text{promo}})}{\sum_i (e_i - s_i)}
\]
where $[s_i, e_i]$ is the $i$-th dense interval and $P_{\text{promo}}$ is the union of promotional periods for the constituent product group.

Table~\ref{tab:window_effect} shows, for each dense interval $[s_i, e_i]$, the number of promotional periods it overlaps.
With $W=47$ days, 83.0\% of the intervals fell within a single promotional period ($\leq 1$ overlap), whereas with $W=141$ days, 64.6\% of the intervals spanned multiple promotional periods ($\geq 2$ overlaps); the mean number of overlapping periods was 1.06 and 1.93 for $W=47$ and $W=141$ days, respectively.

Furthermore, the 149 patterns detected exclusively with $W=47$ days had a mean $R_{\text{promo}}$ of 76.5\%, exceeding the mean $R_{\text{promo}}$ of 58.4\% for patterns commonly extracted under both window sizes.
This indicates that patterns detected only with $W=47$ days are those concentrated exclusively within specific campaign periods.

\begin{table*}[t]
\centering
\caption{Effect of window size: comparison between $W=47$ days and $W=141$ days.}
\label{tab:window_effect}
\small
\begin{tabular}{lrr}
\toprule
 & $W=47$ days & $W=141$ days \\
\midrule
Number of extracted patterns & 513 & 364 \\
\midrule
Mean number of intervals per pattern & 23.3 & 9.2 \\
Mean interval length & 222 days & 385 days \\
\midrule
Fraction of intervals overlapping $\leq 1$ promotional period & 83.0\% & 35.4\% \\
Fraction of intervals overlapping $\geq 2$ promotional periods & 17.0\% & 64.6\% \\
Mean number of overlapping promotional periods & 1.06 & 1.93 \\
\bottomrule
\end{tabular}
\end{table*}

Table~\ref{tab:case_promo} shows three examples of high-$R_{\text{promo}}$ dense patterns.
A common trend across all three examples is that with $W=47$ days, multiple dense intervals are extracted individually, each with $R_\text{promo} \geq 94\%$, showing high temporal overlap with campaign periods.
With $W=141$ days, on the other hand, the dense intervals are merged into a single long-term interval around DAY $[550, 711]$, collectively capturing multiple consecutive campaigns.

\begin{table*}[!t]
\centering
\caption{Examples of high-$R_{\text{promo}}$ dense patterns.}
\label{tab:case_promo}
\footnotesize
\setlength{\tabcolsep}{4pt}
\begin{tabularx}{\linewidth}{p{2.5cm}lp{1.8cm}Xr}
\toprule
Pattern & Setting & Dense Interval & Campaigns overlapping with the dense interval & $R_\text{promo}$ \\
\midrule
\multirow{4}{=}{BAKED BREAD \& PNT BTR/ JELLY/JAMS}
 & \multirow{3}{*}{$W=47$ days} & $[245, 332]$ & C26 $[224, 264]$, C28 $[259, 320]$, C29 $[281, 334]$,\newline C30 $[323, 369]$ & \multirow{3}{*}{100\%} \\
 &  & $[491, 574]$ & C11 $[477, 523]$, C12 $[477, 509]$, C13 $[504, 551]$,\newline C14 $[531, 596]$ &  \\
 &  & $[612, 711]$ & C18 $[587, 642]$, C21 $[624, 656]$, C24 $[659, 719]$ &  \\
 & $W=141$ days & $[549, 711]$ & C13 $[504, 551]$, C14 $[531, 596]$, C18 $[587, 642]$,\newline C21 $[624, 656]$, C24 $[659, 719]$ & 98\% \\
\cmidrule{1-5}
\multirow{3}{=}{FROZEN PIZZA \& SOUP}
 & \multirow{2}{*}{$W=47$ days} & $[211, 329]$ & C26 $[224, 264]$, C27 $[237, 300]$, C28 $[259, 320]$,\newline C29 $[281, 334]$, C30 $[323, 369]$ & \multirow{2}{*}{94\%} \\
 &  & $[561, 711]$ & C16 $[561, 593]$, C17 $[575, 607]$, C18 $[587, 642]$,\newline C21 $[624, 656]$, C22 $[624, 656]$, C23 $[646, 684]$,\newline C25 $[659, 691]$ &  \\
 & $W=141$ days & $[542, 711]$ & C13 $[504, 551]$, C16 $[561, 593]$, C17 $[575, 607]$,\newline C18 $[587, 642]$, C21 $[624, 656]$, C22 $[624, 656]$,\newline C23 $[646, 684]$, C25 $[659, 691]$ & 82\% \\
\cmidrule{1-5}
\multirow{3}{=}{CANNED JUICES \& CHEESE \& FLUID MILK}
 & \multirow{2}{*}{$W=47$ days} & $[240, 330]$ & C26 $[224, 264]$, C28 $[259, 320]$, C29 $[281, 334]$,\newline C30 $[323, 369]$ & \multirow{2}{*}{98\%} \\
 &  & $[563, 711]$ & C17 $[575, 607]$, C18 $[587, 642]$, C21 $[624, 656]$,\newline C22 $[624, 656]$, C23 $[646, 684]$, C25 $[659, 691]$ &  \\
 & $W=141$ days & $[550, 711]$ & C13 $[504, 551]$, C17 $[575, 607]$, C18 $[587, 642]$,\newline C21 $[624, 656]$, C22 $[624, 656]$, C23 $[646, 684]$,\newline C25 $[659, 691]$ & 73\% \\
\bottomrule
\end{tabularx}
\end{table*}

Table~\ref{tab:case_stable} shows three examples of low-$R_{\text{promo}}$ dense patterns.
A common characteristic across all three examples is that $R_\text{promo}$ is low, ranging from 21\% to 39\%, and remains at a low level regardless of the window size.
The total coverage by dense intervals exceeds 550 days, indicating that dense intervals are formed throughout the entire data collection period.
All three are combinations of food items centered around beef, and such product groups are expected to co-occur throughout the year regardless of specific promotional periods.
In addition, although the number of intervals decreases substantially with $W=141$ days, the total coverage is maintained and the variation in $R_\text{promo}$ is small.
This suggests that while the granularity of the number and length of intervals varies with window size, patterns with a broad temporal distribution can be consistently extracted.

\begin{table*}[!t]
\centering
\caption{Examples of low-$R_{\text{promo}}$ dense patterns.}
\label{tab:case_stable}
\footnotesize
\setlength{\tabcolsep}{5pt}
\begin{tabular}{p{3.0cm}lrrr}
\toprule
Pattern & Setting & No.\ of intervals & Total coverage by dense intervals (days) & $R_\text{promo}$ \\
\midrule
\multirow{2}{=}{BEEF \& CHICKEN}
 & $W=47$ days & 54 & 575 & 24\% \\
 & $W=141$ days & 31 & 551 & 11\% \\
\cmidrule{1-5}
\multirow{2}{=}{BEEF \& POTATOES}
 & $W=47$ days & 49 & 557 & 21\% \\
 & $W=141$ days & 9 & 595 & 33\% \\
\cmidrule{1-5}
\multirow{2}{=}{BEEF \& LUNCHMEAT}
 & $W=47$ days & 8  & 644 & 39\% \\
 & $W=141$ days & 1  & 670 & 29\% \\
\bottomrule
\end{tabular}
\end{table*}

In summary, we confirmed that setting the window size to match the promotional period duration enables individual detection of dense intervals corresponding to each promotional period, while a larger window captures long-term dense intervals spanning multiple promotional periods.
Patterns distributed throughout the entire data collection period (711 days) tended to be detected under both window sizes.
The parameter $W$ determines the temporal resolution of the dense intervals; in practice, appropriate selection of $W$ according to the analysis objective is required.

\section{Conclusion}
\label{sec:conclusion}

In this paper, we proposed Apriori-window, an exact algorithm for mining patterns together with their dense intervals from transactional data.
The proposed method is based on direct evaluation of local frequency via a sliding window, and avoids the inherent trade-off between pattern identification accuracy and interval detection accuracy that characterizes occurrence-gap-based methods.
Through candidate interval pruning using the anti-monotonicity of interval support and stride adjustment based on surplus occurrences, we achieved comprehensive and exact dense pattern mining at practical efficiency.
Experiments on three real-world datasets demonstrated that existing methods struggle to simultaneously achieve high accuracy in both pattern identification and dense interval detection, and scalability experiments on synthetic data confirmed that the proposed method operates within 8.9 seconds and 333~MiB of memory even at a scale of one million transactions.
Future directions include application to a wider range of real-world data, analysis of relationships among dense intervals, and online application to streaming data.

\appendix

\section{Synthetic Data Generation for the Scalability Experiment}
\label{sec:app_synthetic}

Synthetic data were generated to evaluate the scalability of the proposed method.
For background transaction generation, the number of transactions $T$, item universe size $I$, and mean transaction length $B$ were taken as inputs.
Each transaction length was drawn from a normal distribution with mean $B$ and standard deviation $\max(1, B/3)$, rounded to an integer, and clipped to the range $[0, I]$, after which the corresponding number of items were sampled uniformly at random without replacement.
To suppress the extraction of patterns of length $\geq 2$ other than the embedded patterns, the expected item occurrence frequency in the background data was adjusted so that co-occurrences other than the embedded patterns were unlikely.

Timestamps were assigned one per transaction; outside the embedding interval, inter-transaction gaps followed a discrete uniform distribution over $[5, 10]$, and inside the embedding interval, gaps were adjusted to satisfy the prescribed interval length.

Dense patterns were embedded with the following parameters:
\begin{itemize}
  \item Pattern length: 5
  \item Number of patterns: 50
  \item Dense interval length: 10,000
  \item Number of dense intervals per pattern: 1
  \item Occurrences per pattern within the interval: 100
\end{itemize}

The total number of patterns of length $\geq 2$ to be extracted was controlled to 1,300.

\section{Data Preprocessing for Experiment D}
\label{sec:app_preprocess}

Each basket in the Dunnhumby Complete Journey dataset is annotated with a day-level timestamp (DAY); however, Apriori-window assumes consecutive integer timestamps.
We therefore sorted baskets in ascending order of (DAY, BASKET\_ID) and defined the timestamp of each basket as
\[
  \text{timestamp} = \text{DAY} \times 1000 + \left\lfloor \frac{\text{rank\_within\_day} \times 1000}{n_{\text{day}}} \right\rfloor
\]
(where $n_{\text{day}}$ is the total number of baskets on that day, at most $668 < 1000$).
This uniform placement distributes baskets evenly across slots 0--999 within each day, so that $W = d \times 1000$ corresponds precisely to a $d$-day window.
The original DAY value can be recovered directly as $\lfloor\text{timestamp}/1000\rfloor$.

\section{Best F1 Parameters for Each Method}
\label{sec:best_params}

We list below the F1-maximizing parameters for each method and setting reported in Table~\ref{tab:combined}.

\paragraph{\textbf{FIM}}
For all datasets, $\mathit{minSup}$ is set to the experimental minimum support $\sigma_{\text{exp}}$.

\paragraph{\textbf{LPFIM}}
\begin{itemize}
  \item Retail:
  \begin{itemize}
    \item $W=250, \sigma=25$: $\sigma_\ell=15\%$, $\mathit{minthd1}=10$, $\mathit{minthd2}=20$
    \item $W=500, \sigma=50$: $\sigma_\ell=15\%$, $\mathit{minthd1}=10$, $\mathit{minthd2}=30$
    \item $W=750, \sigma=75$: $\sigma_\ell=20\%$, $\mathit{minthd1}=10$, $\mathit{minthd2}=20$
    \item $W=1000, \sigma=100$: $\sigma_\ell=20\%$, $\mathit{minthd1}=10$, $\mathit{minthd2}=30$
  \end{itemize}
  \item Chicago:
  \begin{itemize}
    \item $W=250, \sigma=25$: $\sigma_\ell=10\%$, $\mathit{minthd1}=20$, $\mathit{minthd2}=90$
    \item $W=500, \sigma=50$: $\sigma_\ell=10\%$, $\mathit{minthd1}=35$, $\mathit{minthd2}=200$
    \item $W=750, \sigma=75$: $\sigma_\ell=10\%$, $\mathit{minthd1}=45$, $\mathit{minthd2}=290$
    \item $W=1000, \sigma=100$: $\sigma_\ell=10\%$, $\mathit{minthd1}=25$, $\mathit{minthd2}=300$
  \end{itemize}
  \item OnlineRetail:
  \begin{itemize}
    \item $W=100, \sigma=10$: $\sigma_\ell=15\%$, $\mathit{minthd1}=10$, $\mathit{minthd2}=10$
    \item $W=250, \sigma=25$: $\sigma_\ell=15\%$, $\mathit{minthd1}=40$, $\mathit{minthd2}=10$
    \item $W=500, \sigma=50$: $\sigma_\ell=10\%$, $\mathit{minthd1}=55$, $\mathit{minthd2}=100$
  \end{itemize}
\end{itemize}

\paragraph{\textbf{LPPM}}
\begin{itemize}
  \item Retail:
  \begin{itemize}
    \item $W=250, \sigma=25$: $\mathit{maxPer}=10$, $\mathit{minDur}=165$, $\mathit{maxSoPer}=75$
    \item $W=500, \sigma=50$: $\mathit{maxPer}=10$, $\mathit{minDur}=425$, $\mathit{maxSoPer}=105$
    \item $W=750, \sigma=75$: $\mathit{maxPer}=10$, $\mathit{minDur}=635$, $\mathit{maxSoPer}=160$
    \item $W=1000, \sigma=100$: $\mathit{maxPer}=10$, $\mathit{minDur}=685$, $\mathit{maxSoPer}=145$
  \end{itemize}
  \item Chicago:
  \begin{itemize}
    \item $W=250, \sigma=25$: $\mathit{maxPer}=10$, $\mathit{minDur}=145$, $\mathit{maxSoPer}=125$
    \item $W=500, \sigma=50$: $\mathit{maxPer}=10$, $\mathit{minDur}=350$, $\mathit{maxSoPer}=185$
    \item $W=750, \sigma=75$: $\mathit{maxPer}=15$, $\mathit{minDur}=335$, $\mathit{maxSoPer}=10$
    \item $W=1000, \sigma=100$: $\mathit{maxPer}=15$, $\mathit{minDur}=460$, $\mathit{maxSoPer}=15$
  \end{itemize}
  \item OnlineRetail:
  \begin{itemize}
    \item $W=100, \sigma=10$: $\mathit{maxPer}=5$, $\mathit{minDur}=25$, $\mathit{maxSoPer}=90$
    \item $W=250, \sigma=25$: $\mathit{maxPer}=5$, $\mathit{minDur}=50$, $\mathit{maxSoPer}=55$
    \item $W=500, \sigma=50$: $\mathit{maxPer}=10$, $\mathit{minDur}=205$, $\mathit{maxSoPer}=120$
  \end{itemize}
\end{itemize}

\paragraph{\textbf{PFPM}}
\begin{itemize}
  \item Retail:
  \begin{itemize}
    \item $W=250, \sigma=25$: $\mathit{maxPer}=885$, $\mathit{minSup}=25$
    \item $W=500, \sigma=50$: $\mathit{maxPer}=255$, $\mathit{minSup}=50$
    \item $W=750, \sigma=75$: $\mathit{maxPer}=255$, $\mathit{minSup}=75$
    \item $W=1000, \sigma=100$: $\mathit{maxPer}=255$, $\mathit{minSup}=100$
  \end{itemize}
  \item Chicago:
  \begin{itemize}
    \item $W=250, \sigma=25$: $\mathit{maxPer}=6050$, $\mathit{minSup}=25$
    \item $W=500, \sigma=50$: $\mathit{maxPer}=3655$, $\mathit{minSup}=50$
    \item $W=750, \sigma=75$: $\mathit{maxPer}=3655$, $\mathit{minSup}=75$
    \item $W=1000, \sigma=100$: $\mathit{maxPer}=3560$, $\mathit{minSup}=100$
  \end{itemize}
  \item OnlineRetail:
  \begin{itemize}
    \item $W=100, \sigma=10$: $\mathit{maxPer}=3400$, $\mathit{minSup}=10$
    \item $W=250, \sigma=25$: $\mathit{maxPer}=1165$, $\mathit{minSup}=25$
    \item $W=500, \sigma=50$: $\mathit{maxPer}=785$, $\mathit{minSup}=50$
  \end{itemize}
\end{itemize}

\paragraph{\textbf{PPFPM}}
\begin{itemize}
  \item Retail:
  \begin{itemize}
    \item $W=250, \sigma=25$: $\mathit{maxPer}=50$, $\mathit{minPR}=0.90$, $\mathit{minSup}=25$
    \item $W=500, \sigma=50$: $\mathit{maxPer}=50$, $\mathit{minPR}=0.95$, $\mathit{minSup}=50$
    \item $W=750, \sigma=75$: $\mathit{maxPer}=50$, $\mathit{minPR}=0.95$, $\mathit{minSup}=75$
    \item $W=1000, \sigma=100$: $\mathit{maxPer}=50$, $\mathit{minPR}=0.95$, $\mathit{minSup}=100$
  \end{itemize}
  \item Chicago:
  \begin{itemize}
    \item $W=250, \sigma=25$: $\mathit{maxPer}=15$, $\mathit{minPR}=0.30$, $\mathit{minSup}=25$
    \item $W=500, \sigma=50$: $\mathit{maxPer}=185$, $\mathit{minPR}=0.95$, $\mathit{minSup}=50$
    \item $W=750, \sigma=75$: $\mathit{maxPer}=45$, $\mathit{minPR}=0.75$, $\mathit{minSup}=75$
    \item $W=1000, \sigma=100$: $\mathit{maxPer}=45$, $\mathit{minPR}=0.75$, $\mathit{minSup}=100$
  \end{itemize}
  \item OnlineRetail:
  \begin{itemize}
    \item $W=100, \sigma=10$: $\mathit{maxPer}=105$, $\mathit{minPR}=0.55$, $\mathit{minSup}=10$
    \item $W=250, \sigma=25$: $\mathit{maxPer}=70$, $\mathit{minPR}=0.65$, $\mathit{minSup}=25$
    \item $W=500, \sigma=50$: $\mathit{maxPer}=80$, $\mathit{minPR}=0.80$, $\mathit{minSup}=50$
  \end{itemize}
\end{itemize}

\paragraph{\textbf{RPM}}
\begin{itemize}
  \item Retail:
  \begin{itemize}
    \item $W=250, \sigma=25$: $\mathit{maxPer}=30$, $\mathit{minPS}=25$, $\mathit{minRec}=1$
    \item $W=500, \sigma=50$: $\mathit{maxPer}=40$, $\mathit{minPS}=50$, $\mathit{minRec}=1$
    \item $W=750, \sigma=75$: $\mathit{maxPer}=45$, $\mathit{minPS}=75$, $\mathit{minRec}=1$
    \item $W=1000, \sigma=100$: $\mathit{maxPer}=45$, $\mathit{minPS}=100$, $\mathit{minRec}=1$
  \end{itemize}
  \item Chicago:
  \begin{itemize}
    \item $W=250, \sigma=25$: $\mathit{maxPer}=30$, $\mathit{minPS}=25$, $\mathit{minRec}=1$
    \item $W=500, \sigma=50$: $\mathit{maxPer}=35$, $\mathit{minPS}=50$, $\mathit{minRec}=1$
    \item $W=750, \sigma=75$: $\mathit{maxPer}=35$, $\mathit{minPS}=75$, $\mathit{minRec}=1$
    \item $W=1000, \sigma=100$: $\mathit{maxPer}=45$, $\mathit{minPS}=100$, $\mathit{minRec}=1$
  \end{itemize}
  \item OnlineRetail:
  \begin{itemize}
    \item $W=100, \sigma=10$: $\mathit{maxPer}=40$, $\mathit{minPS}=10$, $\mathit{minRec}=1$
    \item $W=250, \sigma=25$: $\mathit{maxPer}=75$, $\mathit{minPS}=25$, $\mathit{minRec}=1$
    \item $W=500, \sigma=50$: $\mathit{maxPer}=80$, $\mathit{minPS}=50$, $\mathit{minRec}=1$
  \end{itemize}
\end{itemize}

\end{document}